\documentclass[letterpaper,11pt]{article}
\usepackage[margin=1.0in]{geometry}
\usepackage{natbib}
\usepackage{amsmath,amsthm,amssymb,epsfig,endnotes}
\usepackage{graphicx}
\graphicspath{ {images/} }
\usepackage{indentfirst}
\usepackage{setspace}
\usepackage{pdflscape}
\usepackage[linesnumbered,ruled]{algorithm2e}

\newcommand{\bbeta}{\boldsymbol{\beta}}
\newcommand{\bgamma}{\boldsymbol{\gamma}}
\newcommand{\bepsilon}{\boldsymbol{\varepsilon}}
\newcommand{\bmu}{\boldsymbol{\mu}}
\newcommand{\gxxg}{\mathbf{G}\mathbf{X}^T\mathbf{X}\mathbf{G}}
\newcommand{\xtx}{\mathbf{X}^T\mathbf{X}}
\newcommand{\invmat}{\left(\xtx+\frac{1}{\kappa}\gxxg\right)^{-1}}
\newcommand{\xjtxj}{\mathbf{x}_j^T\mathbf{x}_j}
\newcommand{\xjty}{\mathbf{x}_j^T\mathbf{y}}

\newtheorem{lemma}{Lemma}[section]
\newtheorem{thm}{Theorem}[section]
\newtheorem{coro}{Corollary}[section]

\begin{document}
\date{}

\title{Bayesian Variable Selection for Linear Regression\\ with the $\kappa$-$\mathbf{G}$ Priors}

\author{
{\rm Zichen Ma}\\
Department of Statistics\\
University of South Carolina\\
\text{zichen@email.sc.edu}\\
\and
{\rm Ernest Fokou\'e}\\
School of Mathematical Science\\
Rochester Institute of Technology\\
\text{epfeqa@rit.edu}
}

\maketitle

\thispagestyle{empty}
\numberwithin{equation}{section}

\subsection*{Abstract}
In this paper, we introduce a new methodology for Bayesian variable selection in linear regression 
that is independent of the traditional indicator method. A diagonal matrix $\mathbf{G}$ is introduced 
to the prior of the coefficient vector $\bbeta$, with each of the $g_j$'s, bounded between 
$0$ and $1$, on the diagonal serves as a stabilizer of the corresponding $\beta_j$. Mathematically, a 
promising variable has a $g_j$ value that is close to $0$, whereas the value of $g_j$ corresponding to 
an unpromising variable is close to $1$. This property is proven in this paper under orthogonality together 
with other asymptotic properties. Computationally, the sample path of each $g_j$ is obtained through 
Metropolis-within-Gibbs sampling method. Also, in this paper we give two simulations to verify the capability 
of this methodology in variable selection.

Keywords: multiple linear regression; Bayesian variable selection; $g$-prior

\doublespacing

\section{Introduction}\label{sec1}

Consider the traditional multiple linear regression (MLR) model having the form
\begin{equation}
\mathbf{y} = \mathbf{X}\bbeta+\bepsilon,~\bepsilon\sim\mathcal{N}(\mathbf{0},\sigma^2\mathbf{I}),
\label{eqn:mlr}
\end{equation}
where $\mathbf{y}$ is an $n\times 1$ response vector, $\mathbf{X}$ an $n\times p$ 
data matrix, $\bbeta$ a $p\times 1$ coefficient vector and $\bepsilon$ 
the random error. In quite a few areas where the linear model applies, an intreresting 
yet very important fact is that only a small portion of variables affect the response 
whereas others are trivial~\citep{jeffreys:1991}. A great many authors have 
discussed this topic from both the frequentist (for example,~\cite{ullah:2013}) 
and the Bayesian perspective~\citep{walli:2011}. In this paper, we proceed 
following the Bayesian path.

In the Bayesian setting (see, for example, \cite{miller:2002} for detail), the coefficient 
$\bbeta$ is usually given a conventional $g$-prior $\mathcal{N}(\mathbf{0},g\sigma^2(\mathbf{X}^T\mathbf{X})^{-1})$, 
introduced in~\cite{zellner:1986}. The g-prior has been given much attention in Bayesian 
variable selection primarily because it leads to a computationally tractable Bayes 
Factor. By introducing an indicator vector, variables are selected and different 
subsets of variables are compared to each other, or to a reference, based on the 
value of Bayes Factor. Multiple works have been done to review this methodology. 
For a recent one, see~\cite{dey:2015}.

In detail, a random indicator vector $\bgamma=(\gamma_1,\gamma_2,\cdots,\gamma_p)^T$ 
is injected to Equation~\eqref{eqn:mlr}, such that for each $\gamma_j, j=1,2,\ldots,p$, 
we have
\begin{equation}
\gamma_j=\begin{cases} 1 & \text{if }\mathbf{x}_j\text{ appears in the model},\\ 0 & \text{otherwise}.\end{cases}
\end{equation}
Thus, for each combination of $\gamma_j$'s, Equation~\eqref{eqn:mlr} is modified to 
\begin{equation}
\nonumber \mathbf{y} = \mathbf{X}_{\bgamma}\bbeta_{\bgamma}+\bepsilon,~\bepsilon\sim\mathcal{N}(\mathbf{0},\sigma^2\mathbf{I}),
\end{equation}
where $\mathbf{X}_{\bgamma}$ is the subset of variables according to 
$\bgamma$ and $\bbeta_{\bgamma}$ is the corresponding coefficient vector. There is a total of $2^p$ 
combinations of $\bgamma$, including the full model, $\bgamma=\mathbf{1}=(1,1,\ldots,1)^T$, and 
the null model, $\bgamma=\mathbf{0}=(0,0,\ldots,0)^T$. For each combination of $\gamma_j$'s, 
a corresponding density $p(\mathbf{y}\mid\bgamma)$ and the Bayes Factor
\begin{equation}
\nonumber \mathbf{BF}_{\bgamma\:\mathbf{1}}=\frac{p(\mathbf{y}\mid\bgamma)}{p(\mathbf{y}\mid\mathbf{1})}.
\end{equation}

A difficulty quickly arises when the dimensionality increases, due to the fact that this method searches 
through the model space of size $2^p$. Certain works have been done to solve this problem. \cite{george:1993}
proposed an empirical method of stochastic search variable selection (SSVS). Each $\beta_j$ is selected 
or rejected based on a Monte Carlo average of $\gamma_j$, coming from a Gibbs-sampler. Such Monte Carlo 
average of $\gamma_j$ is called the posterior inclusion probability (PIP) of $\beta_j$. Similar work 
can be seen in~\cite{barbieri:2004}, in which the authors proposed a median probability model 
rather than a highest probability model, and the variables are selected based on a criterion of $PIP_j>0.5$. 
Further,~\cite{fokoue:2007} modified the method in~\cite{barbieri:2004} to a prevalence model, which 
solved the problem that such median probability model may not exist. Certain works have been doen to 
summarize the Bayesian variable selection with the indicator method.~\cite{ohara:2011} provides 
a thorough review of different methods in Bayesian variable selection.~\cite{han:2001} gives a 
comparison in detail of different empirical Bayes methods, especially the Markov Chain Monte Carlo (MCMC) 
methods, regarding the Bayes Factor.

Certain thoughts have been given to the prior of $\bbeta$ instead of the traditional $g$-prior. 
\cite{george:1997} provides a prior of $\bbeta_{\bgamma}$ that follows 
\begin{equation}
\bbeta_{\bgamma}\sim\mathcal{N}\left(\mathbf{0},\mathbf{D}_{\bgamma}\mathbf{R}_{\bgamma}\mathbf{D}_{\bgamma}\right),
\label{eqn:gandm}
\end{equation}
where $\mathbf{D}_{\bgamma}$ is a diagonal matrix and $\mathbf{R}_{\bgamma}$ is symmetric. Such prior 
gives a good generalization of $g$-prior. \cite{agliari:1988} gives an alternative that follows
\begin{equation}
\bbeta_{\bgamma}\sim\mathcal{N}\left(\mathbf{0},g\sigma^2\left(\mathbf{X}_{\bgamma}^T\mathbf{A}_{\bgamma}\mathbf{X}_{\bgamma}\right)^{-1}\right),
\end{equation}
where $\mathbf{A}_{\bgamma}$ is symmetric and weights different observations, but not the features. Also, 
see~\cite{ley:2001} for a very detailed comparison of different prior choices for Bayesian 
variable selection. Moreover, multiple works have been done to extend the original Zellner's $g$-prior. 
Specifically,~\cite{liang:2008} proposed a study on mixtures of $g$-priors which provides a family of 
hyperpriors on $g$ while still preserves the tractability on the marginal likelihood. \cite{bove:2011} 
developed an extension of the classical Zellner's $g$-prior to generalized linear models, given a large 
family of hyperpriors on $g$.~\cite{maruyama:2011} introduced a fully Bayes formulation with an 
orthogonal decomposition on the matrix $\mathbf{X}_{\bgamma}^T\mathbf{X}_{\bgamma}$, which resolves the 
issue of $p>n$. All the works mentioned above rely on the indicator method, which is classic but somewhat 
redundant. To its worst, the methods still have to face the model space of size $2^p$. In this work, we 
intend to get rid of this indicator method completely.

On the other hand, \cite{tipping:2001} introduced a method called the relevance vector machine (RVM) from the 
machine learning perspective that performs nonparametric variable selection. Retaining the traditional 
Gaussian prior on $\bbeta$, with a little modification, each of the $\beta_j$'s follows a Gaussian prior 
$\mathcal(0,\alpha_j^{-1})$ independently. The parameter $\alpha_j$ serves a purpose as the stabilizer.
That is, since the coefficient $\beta_j$ is \textit{a priori} centered at $0$, the prior variance become 
$0$ as $\alpha_j\rightarrow\infty$, and, on the contrary, the prior of $\beta_j$ becomes flat as $\alpha_j\rightarrow 0$. 
Interestingly, as stated in~\cite{tipping:2004}, combining the non-sparse Gaussian prior on $\bbeta$ with a 
Gamma hyperprior on each of the $\alpha_j$'s, the marginal of $\bbeta$ in fact becomes a multivariate 
t-distribution after integrating out the $\alpha_j$'s, which leads the RVM to a sparse selection machine. 
This property of sparsity is even more elegant when the input in the linear model is raised from feature 
space to kernel space, which is the main focus in~\cite{tipping:2001,tipping:2004}, but not in our work. 

Our work somewhat combines the methodology in~\cite{george:1997} and~\cite{tipping:2004}, but gets 
rid of the traditional indicator method completely. Section~\ref{sec2} provides a thorough theoretical analysis 
on this new method, including the formulation, some important derivation, and some asymptotic properties. 
We introduce the computation of model fitting in~\ref{sec3}. Here we apply the method of Metropolis-within-Gibbs. 
In Section 4, we verify the ability of variable selection of this new methodology with two examples. 
Finally, we provide a summary in Section 5.

\section{The $\kappa$-$\mathbf{G}$ Formulation}\label{sec2}
\subsection{The hierarchical model for variable selection}

Given an MLR model with form~\eqref{eqn:mlr}, we inject a prior to the coefficient $\bbeta$ having the 
form $\mathcal{N}_p(\bbeta\mid\mathbf{0},\kappa\sigma^2(\mathbf{G}\mathbf{X}^T\mathbf{X}\mathbf{G})^{-1})$, 
where, in the variance of the prior, $\kappa>0$ controls the total scale of the variance, and 
$\mathbf{G}=diag(g_1,g_2,\ldots,g_p)$ controls how ``relevant'' each dimension is, with each $g_j\in(0,1)$ 
having an impact to the variance of the corresponding $\beta_j$. This is to some extent a combination 
between~\cite{george:1997} and~\cite{tipping:2001}. In comparison to~\cite{george:1997}, 
the diagonal matrix $\mathbf{D}$ in~\eqref{eqn:gandm} is the matrix $\mathbf{G}^{-1}$ here, and $\mathbf{R}$ 
is the matrix $(\mathbf{X}^T\mathbf{X})^{-1}$. The essential difference is that we have discarded the indicator 
$\bgamma$. Also, in comparison to~\cite{tipping:2001}, this prior can be seen as a parametric analogy to the prior 
given in RVM.

Further, each of the $g_j$'s is assigned an i.i.d. $Beta(a,b)$ prior, and by conjugacy $\kappa$ an 
inverse-gamma prior $IG(\alpha,\theta)$. We keep the setting in~\cite{zellner:1986} for $\sigma^2$, that is, 
a Jeffreys' prior $p(\sigma^2)\propto(\sigma^2)^{-1}$. And thus, the formulation of the hierarchical 
model follows:
\begin{equation}
\begin{split}
\mathbf{y}\left\vert\bbeta,\sigma^2\right. &\sim\mathcal{N}_n\left(\mathbf{y}\left\vert\mathbf{X}\bbeta,\sigma^2\mathbf{I}\right.\right)\\
\bbeta\left\vert\mathbf{G},\kappa,\sigma^2\right. &\sim\mathcal{N}_p\left(\bbeta\left\vert\mathbf{0},\kappa\sigma^2(\gxxg)^{-1}\right.\right)\\
\mathbf{G} &\sim\prod_{j=1}^p Beta\left(g_j\left\vert a,b\right.\right)\\
\kappa &\sim IG(\alpha,\theta)\\
p(\sigma^2) &\propto(\sigma^2)^{-1}
\end{split}
\label{eqn:hier}
\end{equation}
Directly following~\eqref{eqn:hier}, the joint posterior is given by
\begin{equation}
\begin{split}
p\left(\bbeta,\kappa,\mathbf{G},\sigma^2\left\vert\mathbf{y}\right.\right) &\sim p(\mathbf{y}\left\vert\bbeta,\sigma^2\right.)p(\bbeta\left\vert\mathbf{G},\kappa,\sigma^2\right.)p(\kappa)p(\mathbf{G})p(\sigma^2)\\
 &\sim\mathcal{N}_n\left(\mathbf{y}\left\vert\mathbf{X}\bbeta,\sigma^2\mathbf{I}\right.\right)\times\mathcal{N}_p\left(\bbeta\left\vert\mathbf{0},\kappa\sigma^2(\gxxg)^{-1}\right.\right)\\
 &~~~~\times IG(\alpha,\theta)\times\prod_{j=1}^p Beta\left(g_j\left\vert a,b\right.\right)\times(\sigma^2)^{-1}\\
 &\sim \left\vert\sigma^2\mathbf{I}\right\vert^{-1/2}\exp\left\{-\frac{1}{2\sigma^2}\left(\mathbf{y}-\mathbf{X}\bbeta\right)^T\left(\mathbf{y}-\mathbf{X}\bbeta\right)\right\}\times\\
 &~~~~\left\vert\kappa\sigma^2\left(\gxxg\right)^{-1}\right\vert^{-1/2}\exp\left\{-\frac{1}{2\kappa\sigma^2}\bbeta^T\gxxg\bbeta\right\}\\
 &~~~~\times\kappa^{-\alpha-1}\exp\left(-\frac{\theta}{\kappa}\right)\times\left(\prod_{j=1}^p g_j^{a-1}(1-g_j)^{b-1}\right)\times\left(\sigma^2\right)^{-1}.
\end{split}
\label{eqn:joint}
\end{equation}
From~\eqref{eqn:joint}, it is of specific interest to examine the posterior of $\bbeta$ and $\mathbf{G}$. 
The former gives some intuition of the connection between this formulation and both the ordinary least 
square (OLS) estimation and the original Zellner's $g$-prior, whereas the latter is crucial in the understanding 
of variable selection with this model.

\subsection{Posterior of $\bbeta$}

Following~\eqref{eqn:joint}, the posterior of $\bbeta$ is given by
\begin{equation}
\begin{split}
p\left(\bbeta\left\vert\kappa,\mathbf{G},\sigma^2,\mathbf{y}\right\vert\right) &\sim\exp\left\{-\frac{1}{2\sigma^2}\left(\mathbf{y}-\mathbf{X}\bbeta\right)^T\left(\mathbf{y}-\mathbf{X}\bbeta\right)\right\}\times\\
 &~~~~\exp\left\{-\frac{1}{2\kappa\sigma^2}\bbeta^T\gxxg\bbeta\right\}\\
 &\sim\mathcal{N}_p\left(\tilde{\mu}_{\bbeta},\tilde{\Sigma}_{\bbeta}\right),
\end{split}
\label{eqn:postbeta}
\end{equation}
where $\tilde{\mu}_{\bbeta}$ and $\tilde{\Sigma}_{\bbeta}$ are the posterior mean and variance and take 
one the form of 
\begin{equation}
\begin{split}
\tilde{\mu}_{\bbeta} &= \left(\mathbf{X}^T\mathbf{X}+\frac{1}{\kappa}\gxxg\right)^{-1}\mathbf{X}^T\mathbf{y}\\
\tilde{\Sigma}_{\bbeta} &= \sigma^2\left(\mathbf{X}^T\mathbf{X}+\frac{1}{\kappa}\gxxg\right)^{-1}.
\end{split}
\label{eqn:postbetaparam}
\end{equation}
From~\eqref{eqn:postbeta} and~\eqref{eqn:postbetaparam}, we have the following asymptotic results.
\begin{lemma} \label{lemma2.1}
Denote by $\widehat{\bbeta}^{(OLS)}$ the OLS estimator of $\bbeta$. For any $\kappa\neq 0$, as $\mathbf{G}\rightarrow\mathbf{0}$, 
$\tilde{\mu}_{\bbeta}\rightarrow\widehat{\bbeta}^{(OLS)}$ and 
$\tilde{\Sigma}_{\bbeta}\rightarrow Var\left(\widehat{\bbeta}^{(OLS)}\right)$.
\end{lemma}
\begin{proof}
The proof is rather straightforward. Given $\kappa\neq 0$ and $g_j\rightarrow 0$, $\forall j$, 
\begin{equation}
\nonumber \tilde{\mu}_{\bbeta}\rightarrow\left(\mathbf{X}^T\mathbf{X}\right)^{-1}\mathbf{X}^T\mathbf{y}=\widehat{\bbeta}^{(OLS)}
\end{equation}
and 
\begin{equation}
\nonumber \tilde{\Sigma}_{\bbeta}\rightarrow\sigma^2\left(\mathbf{X}^T\mathbf{X}\right)^{-1}=Var\left(\widehat{\bbeta}^{(OLS)}\right).
\end{equation}\qedhere
\end{proof}
\begin{lemma} \label{lemma2.2}
For any $\kappa\neq 0$, as $\mathbf{G}\rightarrow\mathbf{I}$, we have
\begin{equation}
\nonumber \tilde{\mu}_{\bbeta}\rightarrow\frac{\kappa}{\kappa+1}\widehat{\bbeta}^{(OLS)},
\end{equation}
which is the same as the posterior mean of $\bbeta$ in Zellner's g-prior.
\end{lemma}
\begin{proof}
Given $\kappa\neq 0$ and $g_j\rightarrow 1$, $\forall j$,
\begin{equation}
\begin{split}
\nonumber \tilde{\mu}_{\bbeta} &\rightarrow\left(\mathbf{X}^T\mathbf{X}+\frac{1}{\kappa}\mathbf{X}^T\mathbf{X}\right)^{-1}\mathbf{X}^T\mathbf{y}\\
 &=\frac{\kappa}{\kappa+1}\left(\mathbf{X}^T\mathbf{X}\right)^{-1}\mathbf{X}^T\mathbf{y}\\
 &=\frac{\kappa}{\kappa+1}\widehat{\bbeta}^{(OLS)}
\end{split}
\end{equation}\qedhere
\end{proof}

Lemma~\ref{lemma2.1} states that given $\mathbf{G}$ approaches a null matrix, the posterior mean of $\bbeta$ 
approaches the OLS estimator of $\bbeta$. Also notice that $\mathbf{G}\rightarrow\mathbf{0}$ is equivalent 
to assigning a flat prior to $\bbeta$, since the prior would have infinite variance. Thus it would lead 
to a posterior that is equivalent to OLS. Lemma~\ref{lemma2.2} states that in the case where $\mathbf{G}$ 
approaches an identity matrix, the posterior mean of $\bbeta$ converges to the case in the original Zellner's 
$g$-prior, with the parameter $\kappa$ in this formulation being the same as the original parameter $g$. 
This result gives an intuition that the $\kappa$-$\mathbf{G}$ formulation is indeed a generalization of 
Zellner's $g$-prior. Also, it is of interest that as $\kappa\rightarrow\infty$ the convergence from $\tilde{\mu}_{\bbeta}$ 
to $\widehat{\bbeta}^{(OLS)}$ does not require a specific matrix $\mathbf{G}$.

\subsection{Posterior of $\mathbf{G}$}

We then derive the posterior of $\mathbf{G}$ given $\mathbf{y}$, $\kappa$ and $\sigma^2$ by integrating 
out $\bbeta$.
\begin{equation}
\begin{split}
p\left(\mathbf{G}\left\vert\mathbf{y},\sigma^2,\kappa\right.\right) &=\int_{\bbeta}p\left(\mathbf{y}\left\vert\bbeta,\sigma^2\right.\right)p\left(\bbeta\left\vert\mathbf{G}\right.\right)p\left(\mathbf{G}\right)\mathrm{d}\bbeta\\
 &=\prod_{j=1}^p Beta\left(g_j\left\vert a,b\right.\right)\int_{\bbeta}\mathcal{N}_n\left(\mathbf{y}\left\vert\mathbf{X}\bbeta,\sigma^2\mathbf{I}\right.\right)\times\mathcal{N}_p\left(\bbeta\left\vert\mathbf{0},\kappa\sigma^2(\gxxg)^{-1}\right.\right)\mathrm{d}\bbeta\\
 &\propto \left\vert\mathbf{G}\right\vert^a\left\vert\mathbf{I}_p-\mathbf{G}\right\vert^{b-1}\left\vert\xtx+\frac{1}{\kappa}\gxxg\right\vert^{-1/2}\times\\
 &~~~~\exp\left\{\frac{1}{2}\mathbf{y}^T\mathbf{X}\left(\xtx+\frac{1}{\kappa}\gxxg\right)^{-1}\mathbf{X}^T\mathbf{y}\right\}
\end{split}
\label{eqn:postg}
\end{equation}

Unfortunately, the expression in~\eqref{eqn:postg} does not have a closed form. However, we could see 
that the posterior properties of $\mathbf{G}$ relies much on the matrix $\invmat$. And yet we cannot 
proceed the analysis of posterior properties of $\mathbf{G}$ in the most general cases since this inverse 
matrix does not have a further expression in which the matrix $\mathbf{G}$ can be isolated. Figure~\ref{fig:contour} 
gives an intuition of the posterior of $\mathbf{G}$ in the case where $p=2$. Without loss of generality, 
we assume $\mathbf{x}_1$ is a promising variable while $\mathbf{x}_2$ is not. In such case, we have 
$\mathbf{x}_2^T\mathbf{y}=0$ and $\left\vert\mathbf{x}_1^T\mathbf{y}\right\vert\gg 0$. Notice from the 
figure that the posterior of $\mathbf{G}$ is maximized roughly at $g_1\rightarrow 0$ and $g_2\rightarrow 1$. 
This is crucial in linking the $\kappa$-$\mathbf{G}$ formulation and variable selection. Intuitively, 
we would expect a promising variable to have a corresponding $g_j$ close to $0$ while an unpromising variable 
to have a $g_j$ close to $1$.
\begin{figure}[!htbp]
\centering
\includegraphics[width=1.0\textwidth]{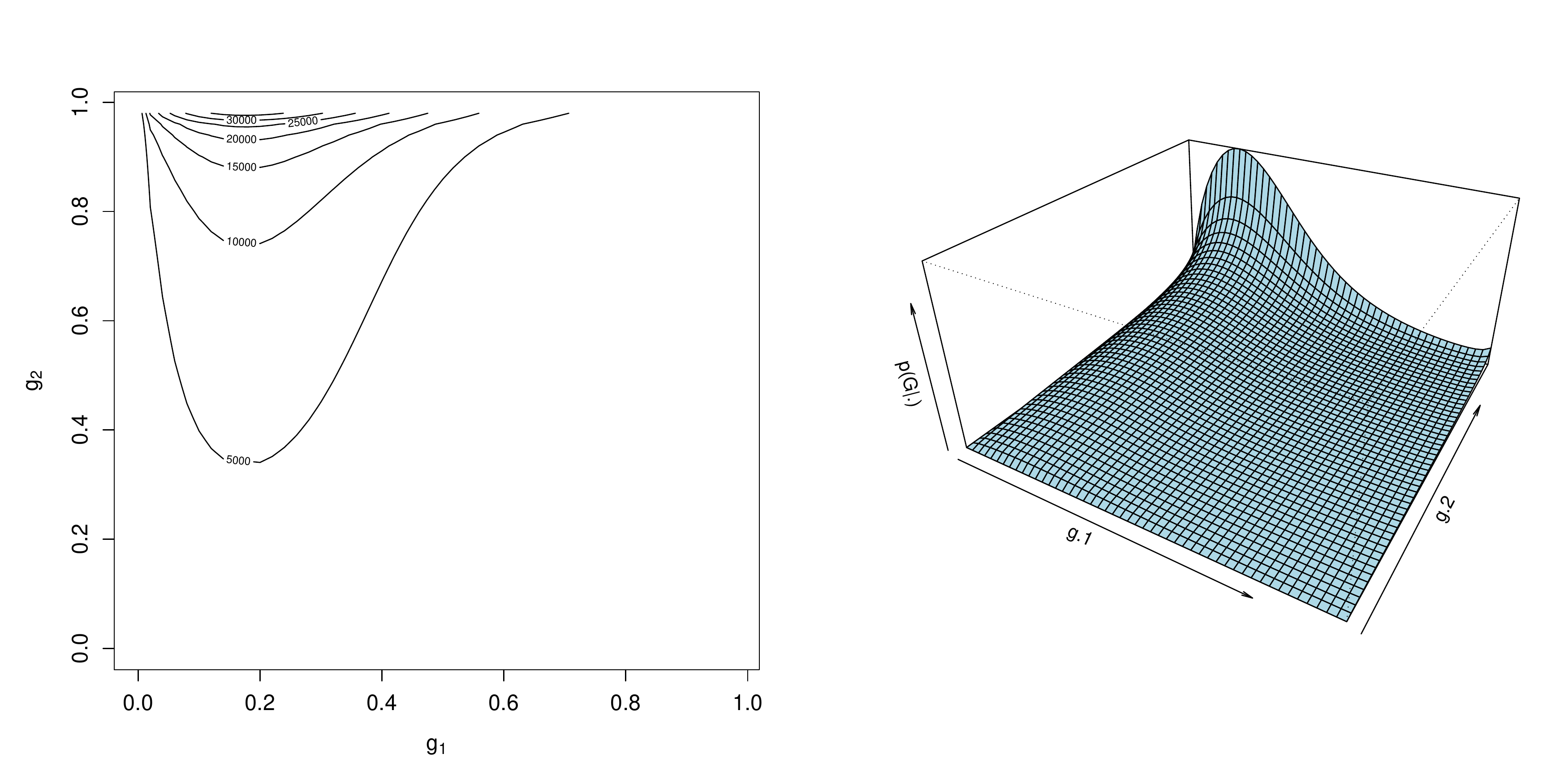}
\caption{Perspective and contour plot of $p\left(\mathbf{G}\left\vert\mathbf{y},\sigma^2,\kappa\right.\right) $}
\label{fig:contour}
\end{figure}

\subsection{A case of orthogonality}
As stated above, much of the posterior properties rely on $\invmat$. Though at this point we are not 
able to proceed to the analysis of the most general case, the analysis under orthogonality where $\xtx=diag(\xjtxj)$ 
is rather tangible. In this case, the posterior of $\mathbf{G}$ in~\eqref{eqn:postg} is simplified to 
\begin{equation}
\begin{split}
p\left(\mathbf{G}\left\vert\mathbf{y},\sigma^2,\kappa\right.\right) &\propto \prod_{j=1}^p \left(g_j^a\left(1-g_j\right)^{b-1}\left(\kappa+g_j^2\right)^{-1/2}\times\right.\\
 &~~~~\left.\exp\left\{\frac{\kappa\left(\xjty\right)^2}{2\sigma^2\xjtxj\left(\kappa+g_j^2\right)}\right\}\right).
\end{split}
\label{eqn:orthog}
\end{equation}
Based on~\eqref{eqn:orthog}, the joint posterior density of $\mathbf{G}$ can be written as the product 
of the marginal posterior density functions of each $g_j$'s, which implies that the $g_j$'s are \textit{a posteriori} 
independent under orthogonality. This simplifies the analysis of $p(\mathbf{G}\vert\cdot)$ by analysing 
each individual posterior density $p(g_j\vert\cdot)$ with 
\begin{equation}
p\left(g_j\left\vert\cdot\right.\right) \propto g_j^a\left(1-g_j\right)^{b-1}\left(\kappa+g_j^2\right)^{-1/2}\exp\left\{\frac{\kappa\left(\xjty\right)^2}{2\sigma^2\xjtxj\left(\kappa+g_j^2\right)}\right\}.
\label{eqn:postgj}
\end{equation}

As was mentioned in the introduction, a crucial question with this formulation is:``how is the $\kappa-\mathbf{G}$ 
methodology linked together with variable selection?'' Such question can be seen in two ways. First, we answer 
how the promising variables lead to certain posterior properties of $g_j$'s. And second, we answer why such 
properties of $g_j$'s indicate certain variables are promising and others are not.
\begin{thm}
A promising variable $\mathbf{x}_j$ has a corresponding $g_j$ that is close to $0$, whereas an unpromising 
variable has a corresponding $g_j$ that is close to $1$.
\end{thm}
\begin{proof}
Without loss of generality, assume $a=b=\frac{1}{2}$ and $\kappa=\sigma^2=1$. Given the posterior 
density $p\left(g_j\vert\mathbf{y},\sigma^2,\kappa\right)$
such that
\begin{equation}
\begin{split}
p\left(g_j\vert\mathbf{y},\sigma^2,\kappa\right) &\propto g_j^{1/2}\left(1-g_j\right)^{-1/2}\left(1+g_j^2\right)^{-1/2}\exp\left\{\frac{\left(\xjty\right)^2}{2\xjtxj\left(1+g_j^2\right)}\right\}\\
 &= g_j^{1/2}\left(1-g_j\right)^{-1/2}\left(1+g_j^2\right)^{-1/2}\exp\left\{\frac{\Vert\mathbf{y}\Vert^2\cos^2\theta_j}{2(1+g_j^2)}\right\},
\end{split}
\label{eqn:thm}
\end{equation}
where $\theta_j$ is the angle between $\mathbf{x}_j$ and $\mathbf{y}$, the general idea of the proof is that 
we find the $g_j$ that maximizes the posterior likelihood, i.e. the \textit{maximum a posteriori} estimate 
for the two cases where $\mathbf{x}_j^T\mathbf{y}=0$ and $\mathbf{x}_j^T\mathbf{y}\neq 0$.

\textit{Unpromising variable.} For an unpromising variable $\mathbf{x}_j$, it is reasonable to assume that 
$\cos\theta_j=0$. Therefore in~\eqref{eqn:thm} $\exp(\cdot)=1$ and we are left with
\begin{equation}
\nonumber p\left(g_j\vert\mathbf{y},\sigma^2,\kappa\right) \propto g_j^{1/2}\left(1-g_j\right)^{-1/2}\left(1+g_j^2\right)^{-1/2},
\label{eqn:thm2}
\end{equation}
which is an increasing function of $g_j$ on $(0,1)$, as $\left(1+g_j^2\right)^{-1/2}$ is monotone decreasing 
from $1$ to $\frac{1}{\sqrt{2}}$, and $g_j^{1/2}\left(1-g_j\right)^{-1/2}$ is monotone 
increasing and $g_j^{1/2}\left(1-g_j\right)^{-1/2}\rightarrow +\infty$ as $g_j\rightarrow 1$. 
Therefore in the case where the variable $\mathbf{x}_j$ is unpromising we have 
\begin{equation}
\widehat{g}_j = \arg\max_{g_j}~p\left(g_j\vert\mathbf{y},\sigma^2,\kappa\right) = 1^-.
\end{equation}

\textit{Promising variable.} For a promising variable $\mathbf{x}_j$, it is reasonable to assume that 
$\cos\theta_j\approx 1$. Since all the terms on the exponent in~\eqref{eqn:thm} are positive, $\exp(\cdot)$ 
is a decreasing function of $g_j$ on $(0,1)$. Further, although the value of $\exp(\cdot)$ somewhat depends 
on $\Vert \mathbf{y}\Vert$, the exponential function dominates the whole posterior likelihood with even 
a moderate value of $\Vert \mathbf{y}\Vert$. Therefore we have
\begin{equation}
\begin{split}
\widehat{g}_j &= \arg\max_{g_j}~p\left(g_j\vert\mathbf{y},\sigma^2,\kappa\right)\\
 &\approx \arg\max_{g_j}~\exp\left\{\frac{\Vert\mathbf{y}\Vert^2\cos^2\theta_j}{2(1+g_j^2)}\right\}\\
 &= 0^+.
\end{split}
\end{equation}
And thus concludes the proof of the theorem.
\end{proof}

Further, Corollary 2.1 provides a very useful result under orthogonality.
\begin{coro}
Under orthogonality, the posterior mean of $\bbeta^{(Bayes)}$ under the $\kappa-\mathbf{G}$ formulation, 
$\tilde{\bmu}_{\bbeta}$, is an unbiased estimator of $\bbeta$.
\end{coro}
\begin{proof}
Denote $\tilde{\mu}_j$ as the posterior mean of the $j$th variable based on the $\kappa-\mathbf{G}$ 
formulation. Under orthogonality, that is, $\xtx=diag(\xjtxj)$, the posterior mean of $\bbeta$ in~\eqref{eqn:postbetaparam} 
is simplified to 
\begin{equation}
\nonumber \tilde{\mu}_j = \frac{\kappa}{\kappa+g_j}\left(\xjtxj\right)^{-1}\xjty = \frac{\kappa}{\kappa+g_j}\widehat{\beta}_j^{(OLS)}.
\end{equation}
As was shown above, we have $g_j\longrightarrow 0$ for a promising variable. Therefore in this case 
\begin{equation}
\nonumber \tilde{\mu}_j \longrightarrow \frac{\kappa}{\kappa+0}\left(\xjtxj\right)^{-1}\xjty = \widehat{\beta}_j^{(OLS)}.
\end{equation}
Since $\widehat{\beta}_j^{(OLS)}$ is an unbiased estimator of $\beta_j$, $\tilde{\mu}_j$ is also unbiased. 

On the other hand, if $g_j\longrightarrow 1$, indicating the variable $\mathbf{x}_j$ does not belong to 
the true model and $\beta_j=0$, the quantity $\widehat{\beta}_j^{(OLS)}$ should capture the unpromising 
feature and converges to $0$ itself. Therefore the bias also vanishes in this case. \qedhere
\end{proof}

\section{Aspects of Computation} \label{sec3}
\subsection{Conditional density of $\sigma^2$ and $\kappa$}

We then introduce the conditional distribution of $\kappa$ and $\sigma^2$, which mostly serve for the computational 
purpose. From~\eqref{eqn:joint}, we obtain a closed-form expression of the conditional density of the 
scale parameter $\kappa$,
\begin{equation}
p\left(\kappa\left\vert\bbeta,\sigma^2,\mathbf{G},\mathbf{y}\right.\right) \sim IG(\tilde{\alpha},\tilde{\theta}),
\label{eqn:postkappa}
\end{equation}
where
\begin{equation}
\begin{split}
\nonumber \tilde{\alpha} &= \frac{p}{2}+\alpha\\
\tilde{\theta} &= \frac{1}{2\sigma^2}\left(\bbeta-\bbeta_0\right)^T\gxxg\left(\bbeta-\bbeta_0\right)+\theta.
\end{split}
\end{equation}

Likewise, the conditional density of $\sigma^2$ also has a closed-form expression given by
\begin{equation}
\begin{split}
p\left(\sigma^2\left\vert\bbeta,\kappa,\mathbf{G},\mathbf{y}\right.\right) \sim& IG\left(\frac{n+p}{2}\right.\\
 & \left.\frac{s^2}{2} + \frac{1}{2}\left(\bbeta-\widehat{\bbeta}\right)^T\xtx\left(\bbeta-\widehat{\bbeta}\right) + \frac{1}{2\kappa}\left(\bbeta-\bbeta_0\right)^T\gxxg\left(\bbeta-\bbeta_0\right)\right),
\end{split}
\label{eqn:postsigma}
\end{equation}
where
\begin{equation}
\begin{split}
\nonumber s^2 &= \left(\mathbf{y}-\mathbf{X}\widehat{\bbeta}\right)^T\left(\mathbf{y}-\mathbf{X}\widehat{\bbeta}\right)\\
\widehat{\bbeta} &= \bbeta^{(OLS)} = \left(\xtx\right)^{-1}\mathbf{X}^T\mathbf{y}.
\end{split}
\end{equation}

\subsection{A useful sampling algorithm}
In this $\kappa$-$\mathbf{G}$ formulation, there are four sets of parameters to be estimated from the data. 
Namely, $\bbeta$ and $\mathbf{G}$, each consisting of $p$ individual parameters, and $\kappa$ and $\sigma^2$. 
The MCMC method is very useful in this case to obtain the sample path of the parameters, and specifically, 
the Gibbs-sampler is a very convenient tool. However, Gibbs-sampler does require the conditional or 
posterior density of the parameters to be known, or of closed-form. As we have addressed before, the exact 
form of the posterior of $\mathbf{G}$ is unknown. Fortunately, the Gibbs sampling of $\mathbf{G}$ can be 
replaced by a Metropolis step, which only requires the density to be known to a proportion. For each draw of 
$\mathbf{G}$, the acceptance ratio is 
\begin{equation}
r = \frac{p\left(\mathbf{G}^{(*)}\left\vert\mathbf{y},\sigma^2,\kappa\right.\right)/J\left(\mathbf{G}^{(*)}\left\vert\mathbf{G}^{(t-1)}\right.\right)}{p\left(\mathbf{G}^{(t-1)}\left\vert\mathbf{y},\sigma^2,\kappa\right.\right)/J\left(\mathbf{G}^{(t-1)}\left\vert\mathbf{G}^{(*)}\right.\right)}
\label{eqn:r}
\end{equation}
where $p(\mathbf{G}\vert\cdot)$ is given by~\eqref{eqn:postg} and $J(\cdot)$ is the proposal distribution 
and is defined as
\begin{equation}
\nonumber J\left(\mathbf{G}^{(*)}\left\vert\mathbf{G}^{(t-1)}\right.\right) \sim \prod_{j=1}^n Beta\left(g_j^{(*)}\left\vert\cdot,\cdot\right.\right).
\end{equation}
Here we assume that the $g_j$'s within each draw are independent. The shape and scale parameters in 
$Beta\left(g_j^{(*)}\left\vert\cdot,\cdot\right.\right)$ may differ in various cases. As any typical Metropolis-Hastings algorithms, 
$\mathbf{G}^{(*)}$ is accepted as $\mathbf{G}^{(t)}$ with probability $min(1,r)$. Thus, the whole 
Metropolis-within-Gibbs algorithm is given in Algorithm~\ref{alg1}.
\begin{algorithm}
\SetKwInOut{Input}{Input}
\Input{data matrix $\mathbf{X}$, response $\mathbf{y}$, initial values $\bbeta^{(0)}$, $(\sigma^2)^{(0)}$, $\kappa^{(0)}$ and $\mathbf{G}^{(0)}$}
\For{$t=1$ \KwTo $T$}{
Update $\bbeta^{(t)}$ based on $p\left(\bbeta^{(t)}\left\vert\kappa^{(t-1)},(\sigma^2)^{(t-1)},\mathbf{G}^{(t-1)},\mathbf{y}\right.\right)$ as in~\eqref{eqn:postbeta}\;
Update $\kappa^{(t)}$ based on $p\left(\kappa^{(t)}\left\vert\bbeta^{(t)},(\sigma^2)^{(t-1)},\mathbf{G}^{(t-1)},\mathbf{y}\right.\right)$ as in~\eqref{eqn:postkappa}\;
Update $(\sigma^2)^{(t)}$ based on $p\left((\sigma^2)^{(t)}\left\vert\bbeta^{(t)},\kappa^{(t)},\mathbf{G}^{(t-1)},\mathbf{y}\right.\right)$ as in~\eqref{eqn:postsigma}\;
Accept $\mathbf{G}^{(t)}=\mathbf{G}^{(*)}$ with probability $min(1,r)$ as in~\eqref{eqn:r}\;
}
\caption{The $\kappa$-$G$ formulation for Bayesian variable selection}
\label{alg1}
\end{algorithm}

Notice that the sampling order, that is, which parameters are updated first each time, is mostly arbitrary. 
We choose to update $\mathbf{G}$ last merely because it involves a Metropolis step, which is more complex 
than the Gibbs steps.

In terms of varaible selection, we would expect the sample path of $g_j$'s of a promising variable to be 
severely skewed to the right within in the support of $(0,1)$, and vice versa. Or in terms of the posterior 
mean of $g_j$, given by
\begin{equation}
\mathbb{E}\left(g_j\vert\mathbf{y}\right) = \widehat{g_j}^{(Bayes)} = \frac{\sum_{t=1}^T g_j^{(t)}}{T},
\label{eqn:expg}
\end{equation}
a promising variable would have a $\widehat{g_j}^{(Bayes)}$ that is close to $0$, and an unpromising 
variable close to $1$.

\section{Numerical Examples and Discussion} \label{sec4}
\subsection{Simulations}

In this section we demonstrate our methodology with two simulated examples. First, consider again when 
$p=2$. $\mathbf{x}_1$ and $\mathbf{x}_2$ both have 30 observations and come from an i.i.d. $\mathcal{N}(0,1)$, 
and the true model is given by
\begin{equation}
\nonumber y_i = 2x_{i1} + \mathcal{N}(0,1).
\end{equation}
Here $\mathbf{x}_1$ is assumed to be the promising variable. Using Algorithm 1, we set the parameters as 
$a=b=0.5$, $\alpha=\theta=1$, and $T=100000$. In the Metropolis step, we use an independent uniform proposal 
distribution
\begin{equation}
\nonumber J\left(\mathbf{G}^{(*)}\left\vert\mathbf{G}^{(t-1)}\right.\right) \sim \prod_{j=1}^n Beta\left(g_j^{(*)}\left\vert 1,1\right.\right).
\end{equation}

Figure~\ref{fig:pequal2} provides a histogram of the sample path of $g_j$'s in the simulation. It is not 
surprising that $g_1$ is severely skewed to the right and concentrates toward $0$, which corresponds to 
$\mathbf{x}_1$ being promising, whereas $g_2$ is severely skewed to the left and concentrates toward $1$, 
corresponding to $\mathbf{x}_2$ being unpromising.
\begin{figure}[!htbp]
\centering
\includegraphics[width=1.0\textwidth]{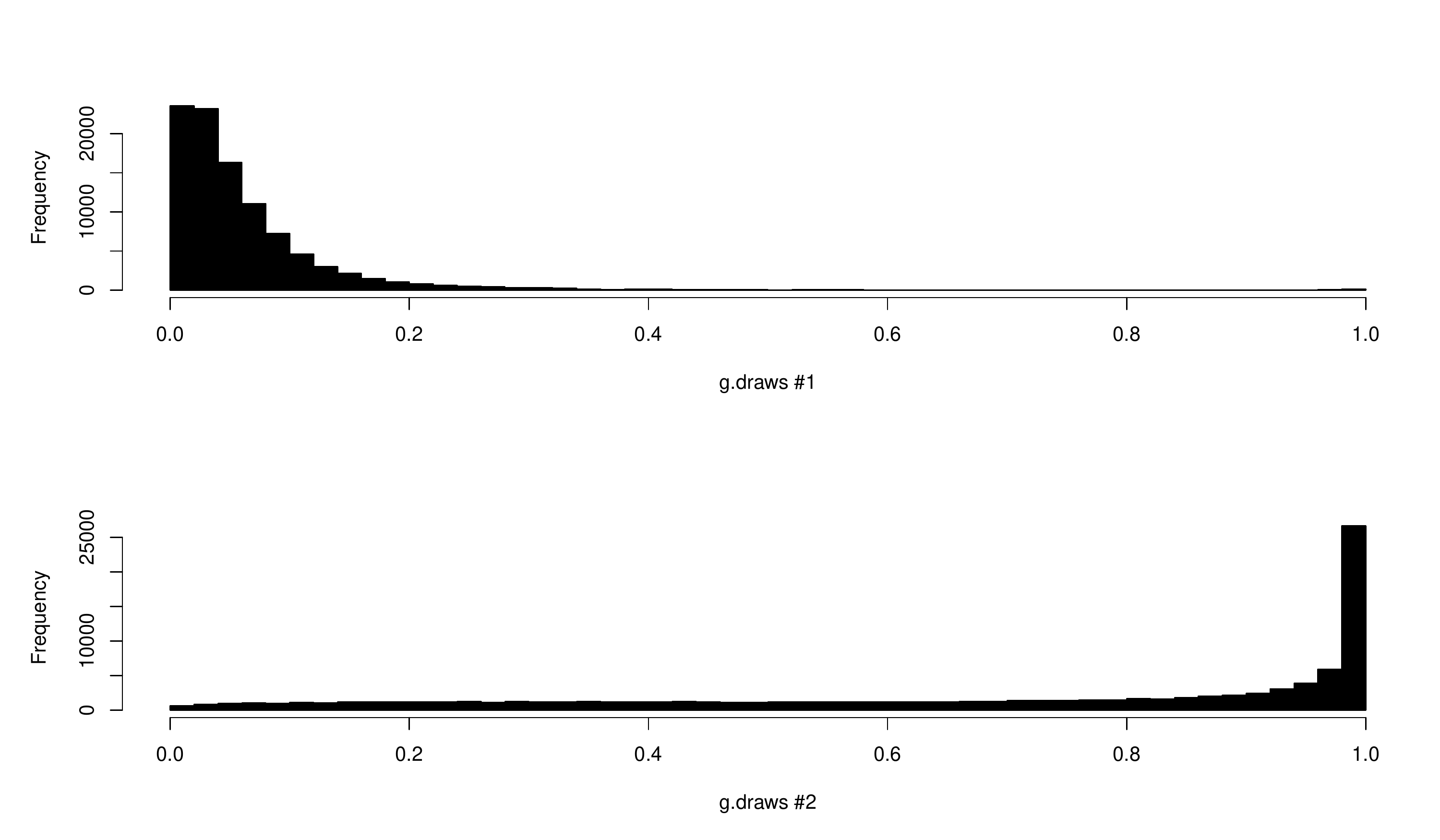}
\caption{Histogram of $g_j$, $p=2$}
\label{fig:pequal2}
\end{figure}
Table~\ref{tbl:pequal2} provides a numerical summary of the $g_j$'s. Due to its severe skewness, here 
we provide both the mean, denoted by $\widehat{g}_j$, and the median, denoted by $\tilde{g}_j$.
\begin{table}[!htbp]
\centering
\caption{Numerical Summary of $g_j$, $p=2$}
\renewcommand\arraystretch{1.5}
\begin{tabular}{c|c|c}
\hline
\textbf{Variable} & $\mathbf{\widehat{g}}_j$ & $\mathbf{\tilde{g}}_j$ \\ \hline
$\mathbf{x}_1$ & $\mathbf{.0682}$ & $\mathbf{.0427}$ \\ \hline
$\mathbf{x}_2$ & $.6986$ & $.8227$ \\ \hline
\end{tabular}
\label{tbl:pequal2}
\end{table}
The numerical summary of $g_j$ for each of the two variable reflects the theoretical deduction in Section 2.

The second example extends the dimensionality mildly to $p=10$. Still, all the predictors are i.i.d. from 
$\mathcal{N}(0,1)$. The true model is given by
\begin{equation}
\nonumber y_i = 2(x_{i1}+x_{i2}+x_{i8}) + \mathcal{N}(0,1).
\end{equation}
The set-up of the algorithm is mostly the same as in the previous example, except that the prior parameters 
of $g_j$ are $a=b=0.3$, instead of $0.5$. In this case, the ``U'' shape of the Beta prior is more strict than 
before. Also we have $T=10000$ in this case. Figure~\ref{fig:pequal10} provides a comparison of the sample 
path of the $g_j$'s. Again, we have $g_1$, $g_2$, and $g_8$ close to 0, which corresponds to the associated 
predictors in the true model.
\begin{figure}[!htbp]
\centering
\includegraphics[width=1.0\textwidth]{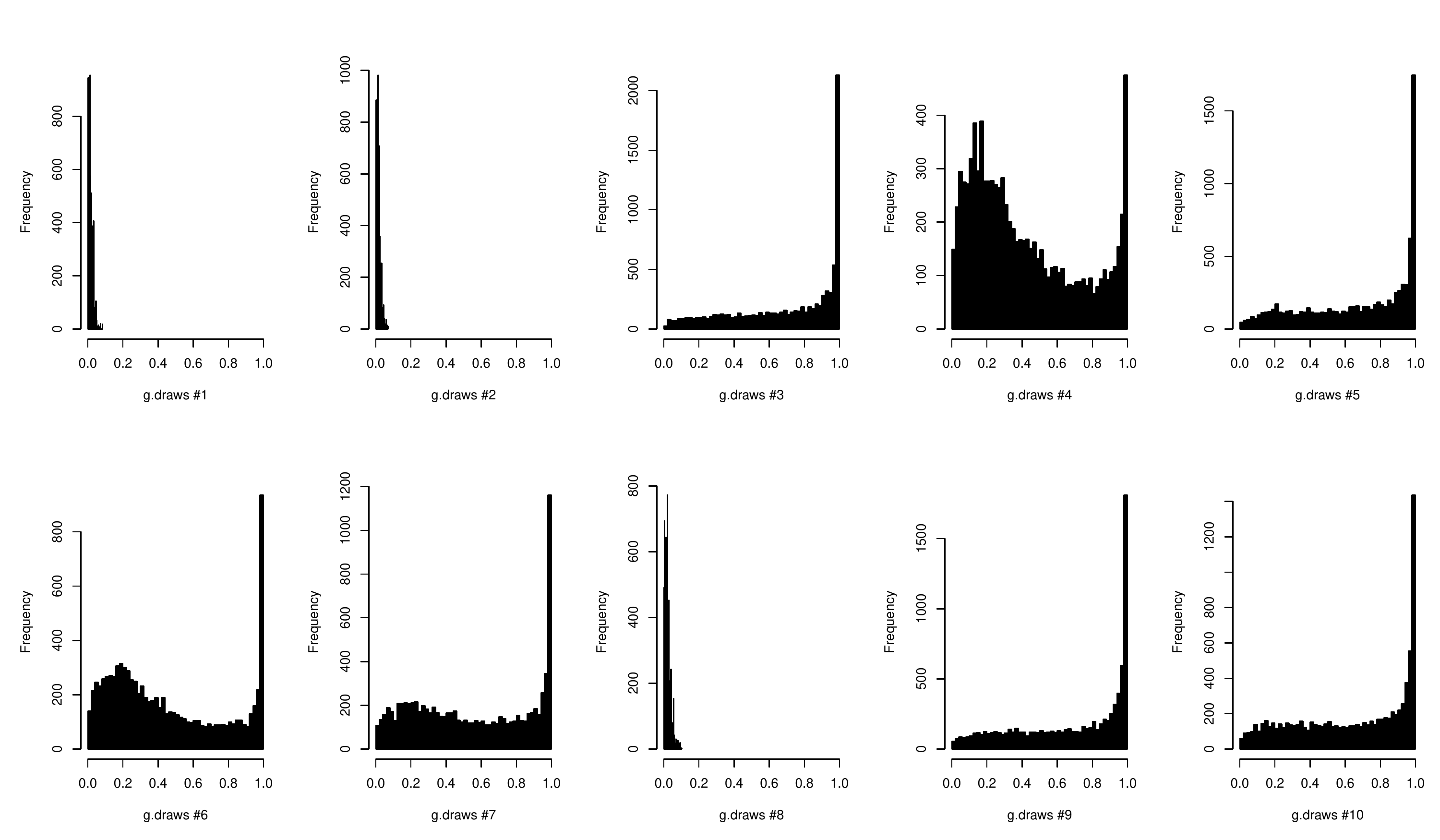}
\caption{Histogram of $g_j$, $p=10$}
\label{fig:pequal10}
\end{figure}

\subsection{Discussion}
In Section~\ref{sec1}, we introduced how this formulation is motivated by the posterior inclusion probability 
(PIP) and the relevance vector machine (RVM). Here we discuss these connections in detail using the 
simulations above. 

As stated before, the value $\widehat{g}_j$ or $\tilde{g}_j$ of a promising variable is close to $0$, so 
that the value of $1-\widehat{g}_j$ or $1-\tilde{g}_j$ is close to $1$. We can see to this quantity $1-\tilde{g}_j$ 
as an analogy to the PIP. However, since the procedure of computing PIP searches the space $\gamma_j\in\{0,1\}$, whereas 
the computation of $g_j$ searches the space $g_j\in(0,1)$, though both quantities are the average of 
their sample path, quite often the PIP equals to $1$ for a promising variable while the value of $\widehat{g}_j$ 
or $\tilde{g}_j$ can hardly be $0$.

Table~\ref{tbl:pequal10} summarizes the quantities $\widehat{g}_j$, $\tilde{g}_j$, $1-\widehat{g}_j$, 
$1-\tilde{g}_j$, and the corresponding PIP in the second simulation. 
\begin{table}[!htbp]
\centering
\caption{Numerical Summary of $g_j$, $p=10$}
\renewcommand\arraystretch{1.5}
\begin{tabular}{c|ccccc||c|ccccc}
\hline
 & $\mathbf{\widehat{g}}_j$ & $1-\mathbf{\widehat{g}}_j$ & $\mathbf{\tilde{g}}_j$ & $1-\mathbf{\tilde{g}}_j$ & $\mathbf{PIP}$ &  & $\mathbf{\widehat{g}}_j$ & $1-\mathbf{\widehat{g}}_j$ & $1-\mathbf{\tilde{g}}_j$ & $\mathbf{\tilde{g}}_j$ & $\mathbf{PIP}$ \\ \hline
$\mathbf{x}_1$ & $\mathbf{.0160}$ & $\mathbf{.9840}$ & $\mathbf{.0130}$ & $\mathbf{.9870}$ & $\mathbf{1.0000}$ & $\mathbf{x}_6$ & $.4517$ & $.5483$ & $.3630$ & $.6370$ & $.0475$ \\
$\mathbf{x}_2$ & $\mathbf{.0150}$ & $\mathbf{.9850}$ & $\mathbf{.0124}$ & $\mathbf{.9876}$ & $\mathbf{1.0000}$ & $\mathbf{x}_7$ & $.5510$ & $.4490$ & $.5364$ & $.4636$ & $.0403$ \\
$\mathbf{x}_3$ & $.6960$ & $.3040$ & $.7952$ & $.2048$ & $.0386$ & $\mathbf{x}_8$ & $\mathbf{.0196}$ & $\mathbf{.9804}$ & $\mathbf{.0157}$ & $\mathbf{.9843}$ & $\mathbf{1.0000}$ \\
$\mathbf{x}_4$ & $.4086$ & $.6914$ & $.3186$ & $.6814$ & $.0775$ & $\mathbf{x}_9$ & $.6762$ & $.3238$ & $.7713$ & $.2287$ & $.0366$ \\
$\mathbf{x}_5$ & $.6749$ & $.3251$ & $.7652$ & $.2348$ & $.0366$ & $\mathbf{x}_{10}$ & $.6352$ & $.3648$ & $.6992$ & $.3008$ & $.0606$ \\ \hline
\end{tabular}
\label{tbl:pequal10}
\end{table}
For the promising variable $\mathbf{x}_1$, $\mathbf{x}_2$, and $\mathbf{x}_8$, $g_j$'s are roughly 
$.01$ while the PIPs equal to $1$, and for the unpromising variables, $g_j$'s are far from $0$ while the 
PIPs are small. Also, it is of interest to notice that the promising variables selected by the two methods 
are identical although the two methodologies are of different origins.

We then consider the connection between the $\kappa$-$\mathbf{G}$ formulation and the relevance vector machine. 
One major similarity between the two is the role of the hyper-parameters. Both $g_j$'s in this paper and 
the $\alpha_j$'s in~\cite{tipping:2001} appear in the prior variance of $\bbeta$. In fact, both $g_j$ 
and $\alpha_j$ serves as the ``stabilizer''. That is, given a Gaussian prior centered at $0$, a large 
value of $g_j$ or $\alpha_j$ yields a high prior precision, or low prior variance of $\beta_j$, so that 
the prior of $\beta_j$ is essentially $0$. However, unlike~\cite{tipping:2001}, in which the prior variance 
of $\beta_j$ is solely $\alpha_j^{-1}$, $g_j$ is only part of the variance, so that it is not necessary 
to set $g_j\in(0,\infty)$, but only a bounded domain between $0$ and $1$ is sufficient. Also, in terms of 
sparsity, the $\kappa$-$\mathbf{G}$ is designed as a sparse machine, that is, we would expect that only 
a few variables affect the response by assigning a ``U-shaped'' Beta hyperprior to the parameter $g_j$.

It is also of interest to verify Corollary 2.1, which indicates, under orthogonality, the unbiasedness of 
$\widehat{\bbeta}^{(Bayes)}$ under this formulation. Table~\ref{tbl:beta} provides a comparison of $\widehat{\bbeta}^{(Bayes)}$ 
and $\widehat{\bbeta}^{(OLS)}$ in the second simulation.
\begin{table}[!htbp]
\centering
\caption{Comparison of $\widehat{\beta}_j^{(Bayes)}$ and $\widehat{\beta}_j^{(OLS)}$}
\renewcommand\arraystretch{1.5}
\begin{tabular}{c|cc||c|cc}
\hline
 & $\kappa$-$\mathbf{G}$ & \textbf{OLS} &  & $\kappa$-$\mathbf{G}$ & \textbf{OLS} \\ \hline
$\widehat{\beta}_1$ & $\mathbf{1.9901}$ & $\mathbf{1.9939}$ & $\widehat{\beta}_6$ & $-.0306$ & $-.0405$ \\ 
$\widehat{\beta}_2$ & $\mathbf{1.9836}$ & $\mathbf{1.9863}$ & $\widehat{\beta}_7$ & $-.0116$ & $-.0192$ \\
$\widehat{\beta}_3$ & $.0089$ & $.0190$ & $\widehat{\beta}_8$ & $\mathbf{1.9517}$ & $\mathbf{1.9545}$ \\
$\widehat{\beta}_4$ & $.0543$ & $.0709$ & $\widehat{\beta}_9$ & $-.0071$ & $-.0124$ \\
$\widehat{\beta}_5$ & $.0022$ & $.0023$ & $\widehat{\beta}_{10}$ & $.0323$ & $.0551$ \\ \hline
\end{tabular}
\label{tbl:beta}
\end{table}
Given the true values as $\beta_1=\beta_2=\beta_8=2$, with $10000$ iterations, the estimates from the methodology 
of this paper are very close to the OLS estimates.

\section{Conclusion}\label{sec5}
In this paper we have demonstrated a new methodology for Bayesian variable selection in linear model that is 
completely independent to the traditional indicator variable method. The coefficient vector $\bbeta$ is given 
a Gaussian prior with the form $\mathcal{N}_p\left(\mathbf{0},\kappa\sigma^2(\gxxg)^{-1}\right)$. By 
injecting a diagonal matrix $\mathbf{G}$ to the variance of the prior, each $g_j$ on the diagonal serves as 
a variance stabilizer such that the promising variables are selected based on the $g_j$'s that are close to 
$0$. Mathematically, under orthogonality, the $g_j$'s are independent and the posterior of each single $g_j$ 
is maximized in the support $(0,1)$ at $g_j\longrightarrow 0$ if the corresponding variable is promising, 
and vice versa. Further, the estimator of $\bbeta$ under orthogonality is asymptotically unbiased. Computationally, 
the hierarchical model is fitted using the Metropolis-within-Gibbs sampling method.

In Section~\ref{sec4}, we have demonstrated through two simulations the usefulness of this methodology 
under orthogonality. Though the dimensionalities in each simulation, $p=2$ and $p=10$ respectively, are 
very mild, the results have shown that this formulation is capable of variable selection and parameter 
estimation, both with considerable accuracy. The systematic or theoretical examination outside orthogonality 
is still remained undone, in which the main difficulty involves the inverse matrix $\invmat$. In conclusion, 
as it is completely independent of searching through the $2^p$ model space, this methodology has the 
potential of selecting variables with higher efficiency comparing to the traditional methodology and 
merits further interest and investigation.

\clearpage
\bibliography{bibliography}

\begin{thebibliography}{}

\bibitem[Agliari and Parisetti, 1988]{agliari:1988}
Agliari, A. and Parisetti, C. (1988).
\newblock {A-g Reference Informative Prior: A Note on Zellner's g-Prior}.
\newblock {\em {Journal of the Royal Statistical Society, Series D}},
  37(3):271--275.

\bibitem[Barbieri and Berger, 2004]{barbieri:2004}
Barbieri, M. and Berger, J. (2004).
\newblock {Optimal Predictive Model Selection}.
\newblock {\em The Annals of Statistics}, 32(3):870--897.

\bibitem[Bov\'e and Held, 2011]{bove:2011}
Bov\'e, D. and Held, L. (2011).
\newblock {Hyper g-Priors for Generalized Linear Models}.
\newblock {\em Bayesian Analysis}, 6(3):387--410.

\bibitem[Dey and Fokou\'e, 2015]{dey:2015}
Dey, T. and Fokou\'e, E. (2015).
\newblock {Bayesian Variable Selection for Predictive Optimal Regression}.
\newblock In {\em {Current Trends in Bayesian Methodology with Applications}}.
  Chapman and Hall.

\bibitem[Fernand\'ez et~al., 2001]{ley:2001}
Fernand\'ez, C., Ley, E., and Steel, M. (2001).
\newblock {Benchmark Priors for Bayesian Model Averaging}.
\newblock {\em Journal of Econometrics}, 100(2):381--427.

\bibitem[Fokou\'e, 2007]{fokoue:2007}
Fokou\'e, E. (2007).
\newblock {Estimation of Atom Prevalence for Optimal Prediction}.
\newblock {\em Contemporary Mathematics}, 447:103--129.

\bibitem[George and McCulloch, 1993]{george:1993}
George, E. and McCulloch, R. (1993).
\newblock {Variable Selection via Gibbs Sampling}.
\newblock {\em Journal of the American Statistical Association}, 85:398--409.

\bibitem[George and McCulloch, 1997]{george:1997}
George, E. and McCulloch, R. (1997).
\newblock {Approaches for Bayesian Variable Selection}.
\newblock {\em Statistical Sinica}.

\bibitem[Han and Carlin, 2001]{han:2001}
Han, C. and Carlin, B. (2001).
\newblock {Markov Chain Monte Carlo Methods for Computing Bayes Factor: A
  Comparative Review}.
\newblock {\em Journal of the American Statistical Association},
  96(455):1122--1132.

\bibitem[Jeffreys and Berger, 1991]{jeffreys:1991}
Jeffreys, W. and Berger, J. (1991).
\newblock {Sharpening Ockham's Razor on a Bayesian Strop}.
\newblock Technical report, University of Texas at Austin, Purdue University.

\bibitem[Liang et~al., 2008]{liang:2008}
Liang, F., Paulo, R., Molina, G., Clyde, M., and Berger, J. (2008).
\newblock {Mixtures of g-Priors for Bayesian Variable Selection}.
\newblock {\em Journal of the American Statistical Association},
  103(481):410--423.

\bibitem[Maruyama and George, 2011]{maruyama:2011}
Maruyama, Y. and George, E. (2011).
\newblock {Fully Bayes Factors with a Generalized g-Prior}.
\newblock {\em The Annals of Statistics}, 39(5):2740--2765.

\bibitem[Miller, 2002]{miller:2002}
Miller, A. (2002).
\newblock {\em {Subset Selection in Regression}}.
\newblock Chapman \& Hall/CRC.

\bibitem[O'Hara and Sillanp\"a\"a, 2011]{ohara:2011}
O'Hara, R. and Sillanp\"a\"a, M. (2011).
\newblock {A Review of Bayesian Variable Selection: What, How, and Which}.
\newblock {\em Bayesian Analysis}, 4(1):85--118.

\bibitem[Tipping, 2001]{tipping:2001}
Tipping, M. (2001).
\newblock {Sparse Bayesian Learning and the Relevance Vector Machine}.
\newblock {\em Journal of Machine Learning Research}, 1:211--244.

\bibitem[Tipping, 2004]{tipping:2004}
Tipping, M. (2004).
\newblock {Bayesian Inference: An Introduction to Principles and Practice in
  Machine Learning}.
\newblock In Bousquet, O., von Luxburg, U., and R\"atsch, G., editors, {\em
  {Advanced Lectures on Machine Learning}}, pages 41--62. Springer.

\bibitem[Ullah and Wang, 2013]{ullah:2013}
Ullah, A. and Wang, H. (2013).
\newblock {Parametric and Nonparametric Frequentist Model Selection and Model
  Averaging}.
\newblock {\em Econometrics}, 1(2):157--179.

\bibitem[Walli and Wagner, 2011]{walli:2011}
Walli, G. and Wagner, H. (2011).
\newblock {Comparing Spike and Slab Priors for Bayesian Variable Selection}.
\newblock {\em Austrian Journal of Statistics}, 40(4):241--264.

\bibitem[Zellner, 1986]{zellner:1986}
Zellner, A. (1986).
\newblock {On Assessing Prior Distributions and Bayesian Regression Analysis
  with g-Prior Distributions}.
\newblock In Goel, P. and Zellner, A., editors, {\em {Bayesian Inference and
  Decision Techniques: Essays in Honor of Bruno de Finetti}}. Amsterdam:
  North-Holland/Elsevier.

\end{thebibliography}
\end{document}